\documentclass[aps,pre,twocolumn]{revtex4-1}

\usepackage[dvipsnames]{xcolor}
\usepackage[ruled,vlined]{algorithm2e}
\usepackage{adjustbox}
\usepackage{algorithmicx}
\usepackage{amssymb}   
\usepackage{bm}        
\usepackage{booktabs}
\usepackage{dcolumn}   
\usepackage{enumitem}
\usepackage{graphicx}
\usepackage{graphicx}  
\usepackage{hyperref}
\usepackage{sidecap}
\usepackage{subfigure}
\usepackage{tabularx}
\usepackage{tabularx}
\usepackage{url}

\usepackage{hyperref}
\definecolor{mylinkcolor}{RGB}{0,0,140}
\hypersetup{colorlinks,allcolors=mylinkcolor,citecolor=mylinkcolor}

\usepackage{paralist}

\renewenvironment{enumerate}[1]{\begin{compactenum}#1}{\end{compactenum}}

\newcommand{\phantomsubfigure}[1]{\subfigure{\label{#1}}}

\newcommand{\gccf}[1]{C_{#1}}   
\newcommand{\lccf}[2]{C_{#1}(#2)}    
\newcommand{\accf}[1]{\bar{C}_{#1}}  
\newcommand{\cliqueL}[2]{K_{#1}(#2)}  
\newcommand{\nbr}[2]{N_{#1}(#2)}  

\newcommand{\onehopnou}{\nbr{1}{u}}

\newcommand{\smallworld}{\textnormal{Small-world}}
\newcommand{\ER}{\textnormal{Erd\H{o}s-R\'enyi}}

\usepackage{tikz,nicefrac,amsmath,pifont,tkz-graph,makecell,tkz-graph}
\usetikzlibrary{arrows,backgrounds,patterns,matrix,shapes,fit,calc,shadows,plotmarks,arrows.meta,positioning,shapes.geometric}

\usepackage{amsthm}

\newtheorem*{theorem*}{Theorem}
\newtheorem*{definition*}{Definition}

\newtheorem{theorem}{Theorem}

\newtheorem{proposition}[theorem]{Proposition}

\newcommand{\hide}[1]{}

\newcommand{\prob}{\mbox{$\mathbb{P}$}}
\newcommand{\probof}[1]{\prob\left[#1\right]}
\newcommand{\expect}[1]{\mbox{$\mathbb{E}$}[#1]}

\newcommand{\expectover}[2]{\mbox{$\mathbb{E}_{#1}$}\left[#2\right]}
\newcommand{\given}{\;\vert\;}

\newcommand{\welldefnodes}{\widetilde V}

\AtBeginDocument{
  \heavyrulewidth=.08em
  \lightrulewidth=.05em
  \cmidrulewidth=.03em
  \belowrulesep=.65ex
  \belowbottomsep=0pt
  \aboverulesep=.4ex
  \abovetopsep=0pt
  \cmidrulesep=\doublerulesep
  \cmidrulekern=.5em
  \defaultaddspace=.5em
}

\makeatletter
\let\Hy@backout\@gobble
\makeatother

\begin{document}


\title{Higher-order clustering in networks}


\author{Hao Yin}
\email[]{yinh@stanford.edu}
\affiliation{Institute for Computational and Mathematical Engineering, Stanford University, Stanford, CA, 94305, USA}

\author{Austin R. Benson}
\email[]{arb@cs.cornell.edu}
\affiliation{Department of Computer Science, Cornell University, Ithaca, NY, 14850, USA}

\author{Jure Leskovec}
\email[]{jure@cs.stanford.edu}
\affiliation{Computer Science Department, Stanford University, Stanford, CA, 94305, USA}


\date{\today}


\begin{abstract}
A fundamental property of complex networks is the tendency for edges to cluster.
The extent of the clustering is typically quantified by the clustering
coefficient, which is the probability that a length-2 path is closed, i.e.,
induces a triangle in the network.
However, higher-order cliques beyond triangles are crucial to understanding
complex networks, and the clustering behavior with respect to such higher-order
network structures is not well understood.
Here we introduce higher-order clustering coefficients that measure the closure
probability of higher-order network cliques and provide a more comprehensive
view of how the edges of complex networks cluster.
Our higher-order clustering coefficients are a natural generalization of the
traditional clustering coefficient.
We derive several properties about higher-order clustering coefficients and
analyze them under common random graph models.
Finally, we use higher-order clustering coefficients to gain new insights into
the structure of real-world networks from several domains.
\end{abstract}


\maketitle


\section{Introduction}

Networks are a fundamental tool for understanding and modeling complex physical,
social, informational, and biological systems~\cite{newman2003structure}.
Although such networks are typically sparse, a recurring trait of networks
throughout all of these domains is the tendency of edges to appear in small
clusters or cliques~\cite{rapoport1953spread,watts1998collective}.  In many
cases, such clustering can be explained by local evolutionary processes.  For
example, in social networks, clusters appear due to the formation of triangles
where two individuals who share a common friend are more likely to become
friends themselves, a process known as \emph{triadic
  closure}~\cite{rapoport1953spread,granovetter1973strength}.  Similar triadic
closures occur in other networks: in citation networks, two references appearing
in the same publication are more likely to be on the same topic and hence more
likely to cite each other~\cite{wu2009modeling} and in co-authorship networks,
scientists with a mutual collaborator are more likely to collaborate in the
future~\cite{jin2001structure}.  In other cases, local clustering arises from
highly connected functional units operating within a larger system, e.g.,
metabolic networks are organized by densely connected
modules~\cite{ravasz2003hierarchical}.

The \emph{clustering coefficient} quantifies the extent to which edges of a
network cluster in terms of triangles.  
The clustering coefficient is defined as the fraction of
length-2 paths, or \emph{wedges}, that are closed with a
triangle~\cite{watts1998collective,barrat2000properties}
(Fig.~\ref{fig:ccf_def}, row $\gccf{2}$).  In other words, the clustering
coefficient measures the probability of triadic closure in the network.
However, the clustering coefficient is inherently restrictive as it measures the
closure probability of just one simple structure---the triangle. Moreover, higher-order
structures such as larger cliques are crucial to the structure and function of
complex networks~\cite{benson2016higher,yaverouglu2014revealing,rosvall2014memory}.
For example, 4-cliques reveal community structure in word association and
protein-protein interaction networks~\cite{palla2005uncovering} and cliques of
sizes 5--7 are more frequent than triangles in many real-world networks with
respect to certain null models~\cite{slater2014mid}.
However, the extent of clustering of such higher-order structures has not been
well understood nor quantified.

\begin{figure}[tb]
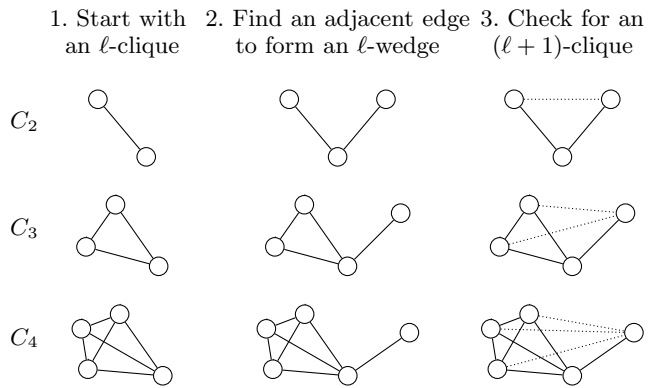

  \begin{tabular}{l c c c}
  & 1.~Start with~ & 2.~Find an adjacent edge & 3.~Check for an \\
  & an $\ell$-clique  & to form an $\ell$-wedge & $(\ell+1)$-clique \\  \\
  $\gccf{2}$ &  \begin{tikzpicture}[baseline=(current bounding box.center)] \input{intro2_part1} \end{tikzpicture} & \begin{tikzpicture}[baseline=(current bounding box.center)] \input{intro2_part1}\input{intro2_part2} \end{tikzpicture} & \begin{tikzpicture}[baseline=(current bounding box.center)] \input{intro2_part1}\input{intro2_part2}\foreach \w in {w1}{ \draw[densely dotted] (\w) -- (v); } \end{tikzpicture} \\  \\
  $\gccf{3}$ &  \begin{tikzpicture}[baseline=(current bounding box.center)] \input{intro3_part1} \end{tikzpicture} & \begin{tikzpicture}[baseline=(current bounding box.center)] \input{intro3_part1}\input{intro3_part2} \end{tikzpicture} & \begin{tikzpicture}[baseline=(current bounding box.center)] \input{intro3_part1}\input{intro3_part2}\foreach \w in {w1,w2}{ \draw[densely dotted] (\w) -- (v); } \end{tikzpicture} \\  \\
  $\gccf{4}$ &  \begin{tikzpicture}[baseline=(current bounding box.center)] \input{intro4_part1} \end{tikzpicture} & \begin{tikzpicture}[baseline=(current bounding box.center)] \input{intro4_part1}\input{intro4_part2} \end{tikzpicture} & \begin{tikzpicture}[baseline=(current bounding box.center)] \input{intro4_part1}\input{intro4_part2}\foreach \w in {w1,w2,w3}{\draw[densely dotted] (\w) -- (v);} \end{tikzpicture} \\        
  \end{tabular}
\caption{Overview of higher-order clustering coefficients as clique
  expansion probabilities.  The $\ell$th-order clustering coefficient
  $\gccf{\ell}$ measures the probability that an $\ell$-clique and an adjacent
  edge, i.e., an $\ell$-wedge, is closed, meaning that the $\ell - 1$ possible
  edges between the $\ell$-clique and the outside node in the adjacent edge
  exist to form an $(\ell + 1)$-clique.}  \label{fig:ccf_def}
\end{figure}  

Here, we provide a framework to quantify higher-order clustering in networks by
measuring the normalized frequency at which higher-order cliques are closed,
which we call \emph{higher-order clustering coefficients}.
We derive our higher-order clustering coefficients by extending a novel
interpretation of the classical clustering coefficient as a form
of clique expansion (Fig.~\ref{fig:ccf_def}).
We then derive several properties about higher-order clustering coefficients
and analyze them under the $G_{n,p}$ and small-world null models.

Using our theoretical analysis as a guide, we analyze the higher-order clustering
behavior of real-world networks from a variety of domains.
We find that each domain of networks has its own higher-order clustering pattern,
which the traditional clustering coefficient does not show on its own.
Conventional wisdom in network science posits that practically all real-world
networks exhibit clustering;
however, we find that not all networks exhibit higher-order clustering.
More specifically, once we control for the clustering as measured by the
classical clustering coefficient, some networks do not show significant
clustering in terms of higher-order cliques.
In addition to the theoretical  properties and empirical findings
exhibited in this paper, our related work also demonstrates a connection
between higher-order clustering and community detection~\cite{yin2017local}.


\section{Derivation of higher-order clustering coefficients}

In this section, we derive our higher-order clustering coefficients and some of their
basic properties.
We first present an alternative interpretation of the classical clustering coefficient
and then show how this novel interpretation seamlessly generalizes to arrive at our 
definition of higher-order clustering coefficients.
We then provide some probabilistic interpretations of higher-order clustering coefficients that will be useful
for our subsequent analysis.

\subsection{Alternative interpretation of the classical clustering coefficient}

Here we give an alternative interpretation of the clustering coefficient that
will later allow us to generalize it and quantify clustering of higher-order
network structures (this interpretation is summarized in Fig.~\ref{fig:ccf_def}). 
Our interpretation is based on a
notion of clique expansion. First, we consider a $2$-clique $K$ in a graph $G$
(that is, a single edge $K$; see Fig.~\ref{fig:ccf_def}, row $C_2$, column 1).
Next, we \emph{expand} the clique $K$ by considering any edge $e$ adjacent to
$K$, i.e., $e$ and $K$ share exactly one node (Fig.~\ref{fig:ccf_def}, row
$C_2$, column 2).  
This expanded subgraph forms a wedge, i.e., a length-$2$ path. 
The classical global clustering coefficient $\gccf{}$ of
$G$ (sometimes called the
transitivity of $G$~\cite{boccaletti2006complex}) is then defined as the
fraction of wedges that are \emph{closed}, meaning that the $2$-clique and
adjacent edge induce a $(2 + 1)$-clique, or a triangle (Fig.~\ref{fig:ccf_def},
row $C_2$, column 3)~\cite{barrat2000properties,luce1949method}.
The novelty of our interpretation of the clustering coefficient is considering it
as a form of clique expansion, rather than as the closure of a length-$2$ path,
which is key to our generalizations in the next section.

Formally, the classical global clustering coefficient is
\begin{equation}\label{eq:global_ccf}
  \gccf{} = \frac{6 \lvert K_3 \rvert}{\lvert W \rvert},
\end{equation}
where $K_3$ is the set of $3$-cliques (triangles), $W$ is the set of wedges, and
the coefficient $6$ comes from the fact that each $3$-clique closes 6 
wedges---the 6 ordered pairs of edges in the triangle.

We can also reinterpret the local clustering coefficient~\cite{watts1998collective}
in this way. In this case, each wedge again consists of a $2$-clique and
adjacent edge (Fig.~\ref{fig:ccf_def}, row $C_2$, column 2), and we call the
unique node in the intersection of the $2$-clique and adjacent edge
the \emph{center} of the wedge.  The \emph{local clustering clustering
coefficient} of a node $u$ is the fraction of wedges
centered at $u$ that are closed:
\begin{equation}\label{eq:local_ccf}
\lccf{}{u} = 
\dfrac{2\lvert K_3(u) \rvert}{\lvert W(u) \rvert},
\end{equation}
where $K_3(u)$ is the set of $3$-cliques containing $u$ and $W(u)$ is the set of
wedges with center $u$ (if $\lvert W(u) \rvert = 0$, we say that $\lccf{}{u}$ is
undefined).  The \emph{average clustering coefficient} $\accf{}$ is the mean of
the local clustering coefficients,
\begin{equation}\label{eq:avg_ccf}
  \accf{} = \frac{1}{\lvert \welldefnodes \rvert}\sum_{u \in \welldefnodes} \lccf{}{u},
\end{equation}
where $\welldefnodes$ is the set of nodes in the network where the local
clustering coefficient is defined.

\subsection{Generalizing to higher-order clustering coefficients}\label{sec:generalization}

Our alternative interpretation of the clustering coefficient, described above as
a form of clique expansion, leads to a natural generalization to higher-order
cliques. Instead of expanding $2$-cliques to $3$-cliques, we expand
$\ell$-cliques to $(\ell + 1)$-cliques (Fig.~\ref{fig:ccf_def}, rows $C_3$ and
$C_4$). Formally, we define an $\ell$-wedge to consist of an $\ell$-clique and
an adjacent edge for $\ell \ge 2$. Then we define the global $\ell$th-order clustering
coefficient $\gccf{\ell}$ as the fraction of $\ell$-wedges that are closed,
meaning that they induce an $(\ell + 1)$-clique in the network.  We can write
this as
\begin{equation}\label{eq:global_ccf_l}
  \gccf{\ell} = \frac{(\ell^2 + \ell)\lvert K_{\ell + 1} \rvert}{\lvert W_{\ell} \rvert},
\end{equation}
where $K_{\ell + 1}$ is the set of $(\ell + 1)$-cliques, and $W_{\ell}$ is the
set of $\ell$-wedges.  The coefficient $\ell^2 + \ell$ comes from the fact that
each $(\ell + 1)$-clique closes that many wedges: each $(\ell+1)$-clique
contains $\ell + 1$ $\ell$-cliques, and each $\ell$-clique contains $\ell$ nodes
which may serve as the center of an $\ell$-wedge. Note that the classical
definition of the global clustering coefficient given in Eq.~\ref{eq:global_ccf}
is equivalent to the definition in Eq.~\ref{eq:global_ccf_l} when $\ell = 2$.

We also define higher-order local clustering coefficients:
\begin{equation} \label{eq:def_lccf}
  \lccf{\ell}{u} = \frac{\ell \lvert K_{\ell + 1}(u) \rvert}{\lvert W_{\ell}(u) \rvert},
\end{equation}
where $K_{\ell + 1}(u)$ is the set of $(\ell + 1)$-cliques containing node $u$,
$W_{\ell}(u)$ is the set of $\ell$-wedges with center $u$ (where the center
is the unique node in the intersection of the $\ell$-clique and adjacent edge comprising the wedge; see Fig.~\ref{fig:ccf_def}), 
and the coefficient $\ell$
comes from the fact that each $(\ell + 1)$-clique
containing $u$ closes that many $\ell$-wedges in $W_{\ell}(u)$.
The $\ell$th-order clustering coefficient of a node is defined for any node that
is the center of at least one $\ell$-wedge, and the average $\ell$th-order
clustering coefficient is the mean of the local clustering coefficients:
\begin{equation}\label{eq:def_accf}
  \accf{\ell} = \frac{1}{\lvert \welldefnodes_\ell \rvert}\sum_{u \in \welldefnodes_\ell} \lccf{\ell}{u},
\end{equation}
where $\welldefnodes_{\ell}$ is the set of nodes that are the centers of at least
one $\ell$-wedge.

To understand how to compute higher-order clustering coefficients, we substitute
the following useful identity
\begin{equation}\label{eq:wedge_identity}
  \lvert W_{\ell}(u) \rvert = \lvert K_{\ell}(u) \rvert \cdot (d_u - \ell + 1),
\end{equation}
where $d_u$ is the degree of node $u$, into Eq.~\ref{eq:def_lccf} to get
\begin{equation}   \label{eq:deriv_lccf}
  \lccf{\ell}{u}
  = \frac{\ell \cdot \lvert K_{\ell+1}(u) \rvert}{(d_u - \ell + 1) \cdot \lvert K_{\ell}(u) \rvert }.
\end{equation}
From Eq.~\ref{eq:deriv_lccf}, it is easy to see that we can compute all local
$\ell$th-order clustering coefficients by enumerating all $(\ell + 1)$-cliques
and $\ell$-cliques in the graph.  The computational complexity of the algorithm
is thus bounded by the time to enumerate $(\ell + 1)$-cliques and
$\ell$-cliques. Using the Chiba and Nishizeki algorithm~\cite{chiba1985arboricity},
the complexity is $O(\ell a^{\ell-2}m)$, where $a$ is the arboricity of the graph,
and $m$ is the number of edges.
The arboricity $a$ may be as large as $\sqrt{m}$, so this algorithm
is only guaranteed to take polynomial time if $\ell$ is a constant.
In general, determining if there exists a single clique with at least $\ell$
nodes is NP-complete~\cite{karp1972reducibility}.

For the global clustering coefficient, note that
\begin{equation}
\lvert W_{\ell} \rvert = \sum_{u \in V} \lvert W_\ell(u) \rvert.
\end{equation}
Thus, it suffices to enumerate $\ell$-cliques (to compute $\lvert W_{\ell} \rvert$ using Eq.~\ref{eq:wedge_identity})
and to count the total number of $\ell$-cliques.  In practice, we use
the Chiba and Nishizeki to enumerate cliques and simultaneously compute
$\gccf{\ell}$ and $\lccf{\ell}{u}$ for all nodes $u$.
This suffices for our clustering analysis with $\ell = 2, 3, 4$ 
on networks with over a hundred million edges in Section~\ref{sec:empirical}.

\subsection{Probabilistic interpretations of higher-order clustering coefficients}

To facilitate understanding of higher-order clustering coefficients and to
aid our analysis in Section~\ref{sec:theoretical}, we
present a few probabilistic interpretations of the quantities.  First, we can
interpret $\lccf{\ell}{u}$ as the probability that a wedge $w$ chosen uniformly
at random from all wedges centered at $u$ is closed:
\begin{equation} \label{eq:prob_interp1}
  \lccf{\ell}{u} = \probof{w \in K_{\ell + 1}(u)}.
\end{equation}
The variant of this interpretation for the classical clustering case of $\ell=2$
has been useful for graph algorithm development~\cite{seshadhri2013triadic}.

For the next probabilistic interpretation, it is useful to analyze the structure of the 1-hop
neighborhood graph $\onehopnou$ of a given node $u$ (not containing node $u$).
The vertex set of $\onehopnou$ is the set of
all nodes adjacent to $u$, and the edge set consists of all edges between neighbors
of $u$, i.e., $ \{ (v, w) \;\vert\; (u, v), (u, w), (v, w) \in E \}$, where $E$ is the edge set of the graph.

Any $\ell$-clique in $G$ containing node $u$ corresponds to a unique
$(\ell - 1)$-clique in $\onehopnou$, and specifically for $\ell = 2$, any edge $(u, v)$
corresponds to a node $v$ in $\onehopnou$.  Therefore, each $\ell$-wedge
centered at $u$ corresponds to an $(\ell-1)$-clique $K$ and one of the
$d_u - \ell + 1$ nodes outside $K$ (i.e., in $\onehopnou \backslash K$).
Thus, Eq.~\ref{eq:deriv_lccf} can be re-written as
\begin{equation}\label{eq:lccf_nbrs}
\frac{\ell \cdot \lvert K_{\ell}(\onehopnou) \rvert}{(d_u - \ell + 1) \cdot \lvert K_{\ell-1}(\onehopnou) \rvert },
\end{equation}
where $K_{k}(\onehopnou)$ denotes the number of $k$-cliques in $\onehopnou$.

If we uniformly at random select an $(\ell - 1)$-clique $K$ from $\onehopnou$
and then also uniformly at random select a node $v$ from $\onehopnou$ outside of
this clique, then $\lccf{\ell}{u}$ is the probability that these $\ell$ nodes
form an $\ell$-clique:
\begin{equation} \label{eq:prob_interp2}
  \lccf{\ell}{u} = \probof{K \cup \{ v \} \in K_{\ell}(\onehopnou)}.
\end{equation}

Moreover, if we condition on observing an $\ell$-clique from this sampling
procedure, then the $\ell$-clique itself is selected uniformly at random from
all $\ell$-cliques in $\onehopnou$.  Therefore, $\lccf{\ell-1}{u} \cdot
\lccf{\ell}{u}$ is the probability that an $(\ell - 1)$-clique and two nodes
selected uniformly at random from $\onehopnou$ form an $(\ell + 1)$-clique.
Applying this recursively gives
\begin{equation}
\prod_{j=2}^{\ell}\lccf{j}{u} = \frac{\lvert K_{\ell}(\onehopnou) \rvert}{{d_u \choose \ell}}.
\end{equation}
In other words, the product of the higher-order local clustering coefficients of
node $u$ up to order $\ell$ is the $\ell$-clique density amongst $u$'s
neighbors.


\section{Theoretical analysis and higher-order clustering in random graph models}\label{sec:theoretical}

We now provide some theoretical analysis of our higher-order clustering coefficients.
We first give some extremal bounds on the values that higher-order clustering coefficients
can take given the value of the traditional (second-order) clustering coefficient.
After, we analyze the values of higher-order clustering coefficients in two common
random graph models---the $G_{n,p}$ and small-world models.
The theory from this section will be a useful guide for interpreting the clustering
behavior of real-world networks in Section~\ref{sec:empirical}.

\subsection{Extremal bounds}

\begin{figure}[tb]
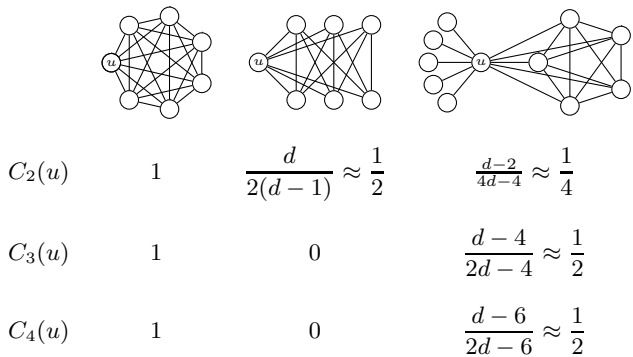

  \begin{tabular}{l @{\hskip 12pt} c @{\hskip 12pt} c @{\hskip 12pt} c}
  & \begin{tikzpicture}[baseline=(current bounding box.center)] \input{ccf_equal} \end{tikzpicture} & \begin{tikzpicture}[baseline=(current bounding box.center)] \input{ccf_lower} \end{tikzpicture} & \begin{tikzpicture}[baseline=(current bounding box.center)] \input{ccf_upper} \end{tikzpicture} \\ \\
  $\lccf{2}{u}$ & $1$ & $\dfrac{d}{2(d - 1)} \approx \dfrac{1}{2}$ & $\frac{d - 2}{4d - 4} \approx \dfrac{1}{4}$  \\  \\
  $\lccf{3}{u}$ & $1$ & $0$ & $\dfrac{d - 4}{2d - 4} \approx \dfrac{1}{2}$ \\ \\
  $\lccf{4}{u}$ & $1$ & $0$ & $\dfrac{d - 6}{2d - 6} \approx \dfrac{1}{2}$    
  \end{tabular}
  \caption{Example 1-hop neighborhoods of a node $u$ with degree $d$ with
  different higher-order clustering. Left: For cliques, $\lccf{\ell}{u} = 1$ for any $\ell$.
  Middle: If $u$'s neighbors form a complete bipartite graph, $\lccf{2}{u}$ is
  constant while $\lccf{\ell}{u} = 0$, $\ell \ge 3$.  Right: If half of $u$'s
  neighbors form a star and half form a clique with $u$, then
  $\lccf{\ell}{u} \approx \sqrt{\lccf{2}{u}}$, which is the upper bound in Proposition~\ref{prop:ccf_bounds}.}
   \label{fig:ccf_diffs}
\end{figure}

We first analyze the relationships between local higher-order clustering
coefficients of different orders.
Our technical result is Proposition~\ref{prop:ccf_bounds}, which
provides essentially tight lower and upper bounds for higher-order local clustering
coefficients in terms of the traditional local clustering coefficient.
The main ideas of the proof are illustrated in Fig.~\ref{fig:ccf_diffs}.

\begin{proposition}\label{prop:ccf_bounds}
  For any fixed $\ell \geq 3$,
\begin{equation}   \label{eq:Bound_Kappa3}
0 \leq \lccf{\ell}{u} \leq \sqrt{\lccf{2}{u}}.
\end{equation}
Moreover,
\begin{enumerate}
\item There exists a finite graph $G$ with a node $u$ such that the lower bound is
  tight and $\lccf{2}{u}$ is within $\epsilon$ of any prescribed value in $[0, \frac{\ell-2}{\ell-1}]$.
\item There exists a finite graph $G$ with a node $u$ such that $\lccf{\ell}{u}$ is within
$\epsilon$ of the upper bound for any prescribed value of $\lccf{2}{u} \in [0, 1]$.
\end{enumerate}
\end{proposition}
\begin{proof}
Clearly, $0 \leq \lccf{\ell}{u}$ if the local clustering coefficient is well
defined. This bound is tight when $\onehopnou$ is $(\ell - 1)$-partite, as 
in the middle column of Fig.~\ref{fig:ccf_diffs}.
In the $(\ell - 1)$-partite case, $\lccf{2}{u} = \frac{\ell-2}{\ell-1}$.
By removing edges from this extremal case in a sufficiently large graph,
we can make $\lccf{2}{u}$ arbitrarily close to any value in $[0, \frac{\ell-2}{\ell-1}]$.

To derive the upper bound, consider the 1-hop neighborhood $\onehopnou$, and let
\begin{equation}
\delta_\ell(\onehopnou) = \frac{\lvert K_{\ell}(\onehopnou) \rvert}{{d_u \choose \ell}}
\end{equation}
denote the $\ell$-clique density of $\onehopnou$.  The Kruskal-Katona
theorem~\cite{kruskal1963number,katona1966theorem} implies that
\begin{align*}
&\delta_{\ell}(\onehopnou) \leq [\delta_{\ell - 1}(\onehopnou)]^{\ell / (\ell - 1)} \\
&\delta_{\ell - 1}(\onehopnou) \leq [\delta_{2}(\onehopnou)]^{(\ell - 1) / 2}.
\end{align*}
Combining this with Eq.~\ref{eq:deriv_lccf} gives
\begin{align*}
  \lccf{\ell}{u} &\leq [\delta_{\ell - 1}(\onehopnou)]^{\frac{1}{\ell - 1}}
  \leq \sqrt{\delta_{2}(\onehopnou)} = \sqrt{\lccf{2}{u}},
\end{align*}
where the last equality uses the fact that $\lccf{2}{u}$ is the edge density of
$\onehopnou$.

The upper bound becomes tight when $\onehopnou$ consists of a clique and isolated nodes
(Fig.~\ref{fig:ccf_diffs}, right) and the neighborhood is sufficiently large.
Specifically, let $\onehopnou$ consist of a clique of size $c$ and $b$ isolated nodes.  
When $\ell = 2$,
\begin{equation*}
\lccf{\ell}{u} = \frac{{c \choose 2}}{{c + b \choose 2}}
= \frac{(c - 1)c}{(c + b - 1)(c + b)} \to \left(\frac{c}{c + b}\right)^2
\end{equation*}
and by Eq.~\ref{eq:lccf_nbrs}, when $3 \le \ell \le c$,
\begin{align*}
\lccf{\ell}{u} 
&= \frac{\ell \cdot {c \choose \ell}}{(c + b - \ell + 1) \cdot {c \choose \ell - 1}} 
= \frac{c - \ell + 1}{c + b - \ell + 1} \to \frac{c}{c + b}.
\end{align*}

By adjusting the ratio $c / (b + c)$ in $\onehopnou$, we can construct a family
of graphs such that $\lccf{2}{u}$ takes any value in the interval $[0,1]$ as $d_u \to \infty$
and $\lccf{\ell}{u} \to \sqrt{\lccf{2}{u}}$ as $d_u \to \infty$.
\end{proof}

The second part of the result requires the neighborhoods to be sufficiently large in order
to reach the upper bound. However, we will see later that in some real-world data, 
there are nodes $u$ for which $\lccf{3}{u}$ is close to the upper bound $\sqrt{\lccf{2}{u}}$ 
for several values of $\lccf{2}{u}$.

Next, we analyze higher-order clustering coefficients in two common random graph
models: the Erd\H{o}s-R\'enyi model with edge probability $p$ (i.e.,
the $G_{n,p}$ model~\cite{erdos1959random}) and the small-world
model~\cite{watts1998collective}.

\subsection{Analysis for the $G_{n,p}$ model}

Now, we analyze higher-order clustering coefficients in classical
Erd\H{o}s-R\'enyi random graph model, where each edge exists 
independently with probability $p$ (i.e., the $G_{n,p}$
model~\cite{erdos1959random}).  
We implicitly assume that $\ell$ is small in the following analysis so that there should be at
least one $\ell$-wedge in the graph (with high probability and $n$
large, there is no clique of size greater than $(2 + \epsilon)\log n / \log (1 / p)$
for any $\epsilon > 0$~\cite{bollobas1976cliques}).
Therefore, the global and local clustering coefficients are well-defined.

In the $G_{n, p}$ model, we first observe that any $\ell$-wedge is closed
if and only if the $\ell - 1$ possible edges between the $\ell$-clique and the
outside node in the adjacent edge exist to form an $(\ell+1)$-clique.  Each of
the $\ell - 1$ edges exist independently with probability $p$ in the $G_{n, p}$
model, which means that the higher-order clustering coefficients should
scale as $p^{\ell - 1}$. We formalize this in the following proposition.

\begin{proposition}\label{prop:ccf_er}
Let $G$ be a random graph drawn from the $G_{n, p}$ model.
For constant $\ell$,
\begin{enumerate}
\item $\expectover{G}{\gccf{\ell}} = p^{\ell - 1}$
\item $\expectover{G}{\lccf{\ell}{u} \given W_{\ell}(u) > 0} = p^{\ell - 1}$ for any node $u$
\item $\expectover{G}{\accf{\ell}} = p^{\ell - 1}$
\end{enumerate}
\end{proposition}
\begin{proof}
  We prove the first part by conditioning on the set of $\ell$-wedges, $W_{\ell}$:
  \begin{align*}
    \expect{\gccf{\ell}}
    &=\textstyle \expectover{G}{\expectover{W_{\ell}}{\gccf{\ell} \given W_{\ell}}} \\
    &=\textstyle  \expectover{G}{\expectover{W_{\ell}}{\frac{1}{\lvert W_{\ell}\rvert}\sum_{w \in W_{\ell}}\probof{w \text{ is closed}}}} \\
    &=\textstyle \expectover{G}{\expectover{W_{\ell}}{\frac{1}{\lvert W_{\ell}\rvert}\sum_{w \in W_{\ell}}p^{\ell - 1}}} \\
    &=\textstyle \expectover{G}{p^{\ell - 1}} \\
    &=\textstyle p^{\ell - 1}.
  \end{align*}  
  As noted above, the second equality is well defined (with high probability)
  for small $\ell$.  The third equality comes from the fact that any
  $\ell$-wedge is closed if and only if the $\ell - 1$ possible edges between
  the $\ell$-clique and the outside node in the adjacent edge exist to form an
  $(\ell+1)$-clique.

  The proof of the second part is essentially the same, except we condition over
  the set of possible cases where $W_{\ell}(u) > 0$.

  Recall that $\welldefnodes$ is the set of nodes at the center of at least one
  $\ell$-wedge.  To prove the third part, we take the conditional expectation
  over $\welldefnodes$ and use our result from the second part.
\end{proof} 
  
The above results say that the global, local, and average $\ell$th order clustering
coefficients decrease exponentially in $\ell$.
It turns out that if we also condition on the second-order clustering coefficient
having some fixed value, then the higher-order clustering coefficients
still decay exponentially in $\ell$ for the $G_{n,p}$ model.
This will be useful for interpreting the distribution of local clustering coefficients
on real-world networks.

\begin{proposition}\label{prop:ccf_er_cond}
Let $G$ be a random graph drawn from the $G_{n, p}$ model. 
Then for constant $\ell$,
\begin{align*}
& \expectover{G}{\lccf{\ell}{u} \given \lccf{2}{u}, W_{\ell}(u) > 0} \\
&=  \left[\lccf{2}{u} - (1 - \lccf{2}{u}) \cdot O(1 / d_u^2)\right]^{\ell - 1} 
\approx (\lccf{2}{u})^{\ell - 1}.
\end{align*}
\end{proposition}
\begin{proof}
  Similar to the proof of Proposition~\ref{prop:ccf_er_cond}, we look at the conditional
  expectation over $W_{\ell}(u) > 0$:
  \begin{align*}
    &\textstyle \expectover{G}{\lccf{\ell}{u} \given \lccf{2}{u}, W_{\ell}(u) > 0} \\
    &=\textstyle \expectover{G}{\expectover{W_{\ell}(u) > 0}{\lccf{\ell}{u} \given \lccf{2}{u},\; W_{\ell}(u)}} \\
    &=\textstyle \expectover{G}{\expectover{W_{\ell}(u) > 0}{\frac{1}{\lvert W_{\ell}(u)\rvert}\sum_{w \in W_{\ell}(u)}\probof{w \text{ closed} \given \lccf{2}{u}}}}.
  \end{align*}  
  Now, note that $\onehopnou$ has $m = \lccf{2}{u} \cdot {d_u \choose 2}$ edges.
  Knowing that $w \in W_{\ell}(u)$ accounts for ${\ell - 1 \choose 2}$ of these
  edges.  By symmetry, the other $q = m - {\ell - 1 \choose 2}$ edges appear in
  any of the remaining $r = {d_u \choose 2} - {\ell - 1 \choose 2}$ pairs of nodes
  uniformly at random.  There are ${r \choose q}$ ways to place these edges, of
  which ${r - \ell + 1 \choose q - \ell + 1}$ would close the wedge $w$.  Thus,
  \begin{align*}
  &\probof{w \text{ is closed} \given \lccf{2}{u}} \\
  &=\textstyle \frac{{r - \ell + 1 \choose q - \ell + 1}}{{r \choose q}}  
  = \frac{(r - \ell + 1)!q!}{(q - \ell + 1)!r!}
  = \frac{(q - \ell + 2)(q - \ell + 3) \cdots q}{(r - \ell + 2)(r - \ell + 3) \cdots r}.
  \end{align*}
  Now, for any small nonnegative integer $k$,
  \begin{align*}
  \frac{q - k}{r - k} &=\textstyle \frac{
  \lccf{2}{u} \cdot {d_u \choose 2} - {\ell -1 \choose 2} - k
  }{
  {d_u \choose 2} - {\ell - 1 \choose 2} - k
  } \\
  &=\textstyle \lccf{2}{u} - (1 - \lccf{2}{u})\left[\frac{{\ell -1 \choose 2} + k}{{d_u \choose 2} - {\ell - 1 \choose 2} - k}\right]  \\
  &=\textstyle \lccf{2}{u} - (1 - \lccf{2}{u}) \cdot O(1 / d_u^2).
  \end{align*}
(Recall that $\ell$ is constant by assumption, so the big-O notation is appropriate).
The above expression approaches $(\lccf{2}{u})^{\ell - 1}$ when $\lccf{2}{u} \to 1$ 
as well as when $d_u \to \infty$.
\end{proof}

Proposition~\ref{prop:ccf_er_cond} says that even if the second-order local
clustering coefficient is large, the $\ell$th-order clustering coefficient will
still decay exponentially in $\ell$, at least in the limit as $d_u$ grows large.  
By examining higher-order clique closures, this allows us to distinguish
between nodes $u$ whose neighborhoods are ``dense but random"
($\lccf{2}{u}$ is large but $\lccf{\ell}{u} \approx (\lccf{2}{u})^{\ell - 1}$)
or ``dense and structured" ($\lccf{2}{u}$ is large \emph{and} $\lccf{\ell}{u} > (\lccf{2}{u})^{\ell - 1}$).
Only the latter case exhibits higher-order clustering.
We use this in our analysis of real-world networks in Section~\ref{sec:empirical}.

\subsection{Analysis for the small-world model}

We also study higher-order clustering in the small-world random graph
model~\cite{watts1998collective}. The model begins with a ring network where
each node connects to its $2k$ nearest neighbors.  Then, for each node $u$ and
each of the $k$ edges $(u, v)$ with $v$ following $u$ clockwise in the ring,
the edge is rewired to $(u, w)$ with probability $p$, where $w$ is chosen
uniformly at random.

\begin{figure}[tb]
  \begin{centering}
  \includegraphics[width=1.0\columnwidth]{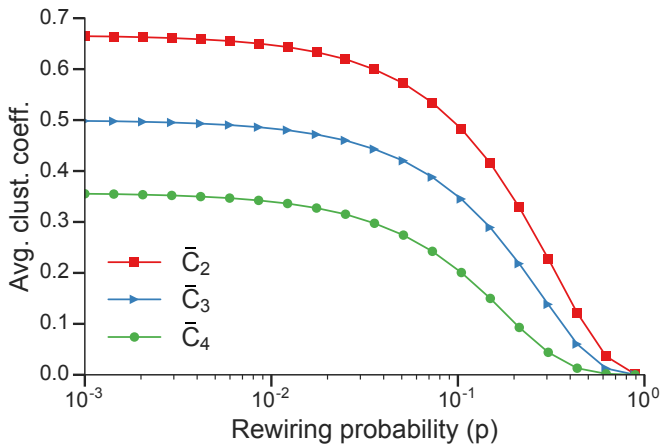}
  \end{centering}
\caption{Average higher-order clustering coefficient $\accf{\ell}$ as a
  function of rewiring probability $p$ in small-world networks for $\ell = 2, 3, 4$
  ($n = 20,000$, $k = 5$). Proposition~\ref{prop:ccf_sw} shows that
  the $\ell$th-order clustering coefficient when $p = 0$ predicts that the clustering should decrease
  modestly as $\ell$ increases.}
  \label{fig:sw_ccfs}
\end{figure}

With no rewiring ($p = 0$) and $k \ll n$, 
it is known that $\accf{2} \approx 3/4$~\cite{watts1998collective}.
As $p$ increases, the average clustering coefficient $\accf{2}$ slightly decreases until a phase transition near $p = 0.1$,
where $\accf{2}$ decays to $0$~\cite{watts1998collective} (also see
Fig.~\ref{fig:sw_ccfs}).
Here, we generalize these results for higher-order clustering coefficients.
\begin{proposition}\label{prop:ccf_sw}
In the small-world model without rewiring ($p = 0$),
\begin{equation*}
  \accf{\ell} \rightarrow (\ell + 1) / (2\ell)
\end{equation*}
for any constant $\ell \geq 2$ as $k \to \infty$ and $n \to \infty$ while $2k < n$.
\end{proposition}
\begin{proof}
Applying  Eq.~\ref{eq:deriv_lccf}, it suffices to show that
\begin{equation}   \label{eq:K_SW}
\lvert K_{\ell}(u) \rvert = \frac{\ell}{(\ell - 1)!}\cdot k^{\ell -1 } + O(k^{\ell - 2})
\end{equation}
as 
\begin{align*}
\lccf{\ell}{u} = \frac{\ell \cdot \frac{(\ell + 1) k^{\ell}}{\ell!}}{(2k - \ell + 1) \cdot \frac{\ell k^{\ell -1 }}{(\ell - 1)!} },
\end{align*}
which approaches $\frac{\ell + 1}{2\ell}$ as $k \to \infty$.

Now we give a derivation of Eq.~\ref{eq:K_SW}. 
We first label the $2k$ neighbors of $u$ as $1, 2, \ldots, 2k$ by their clockwise ordering in the ring.
Since $2k < n$, these nodes are unique.
Next, define the
\emph{span} of any $\ell$-clique containing $u$ as the difference between the
largest and smallest label of the $\ell-1$ nodes in the clique other than $u$.
The span $s$ of any $\ell$-clique satisfies $s \leq k-1$ since any node is
directly connected with a node of label difference no greater than $k-1$.  Also,
$s \geq \ell-2$ since there are $\ell-1$ nodes in an $\ell$-clique other than
$u$.  For each span $s$, we can find $2k-1-s$ pairs of $(i, j)$ such that
$1\leq i$, $j \leq 2k$ and $j - i = s$.  Finally, for every such pair $(i, j)$,
there are ${s - 1 \choose \ell - 3}$ choices of $\ell-3$ nodes between $i$ and
$j$ which will form an $\ell$-clique together with nodes $u$, $i$, and
$j$. Therefore,
\begin{align*}
\lvert \cliqueL{\ell}{u} \rvert
&=\textstyle \sum_{s = \ell-2}^{k-1} (2k-1-s) \cdot {s - 1 \choose \ell - 3} \\
&=\textstyle \sum_{s = \ell-2}^{k-1} (2k-1-s) \cdot \frac{(s-1)(s - 2) \cdots (s - \ell + 3)}{(\ell - 3)!} \\
&=\textstyle \sum_{t = 1}^{k-\ell+2} (2k + 2 -t - \ell) \cdot \frac{t (t + 1) \cdots (t + \ell - 4)}{(\ell - 3)!}.
\end{align*}
If we ignore lower-order terms $k$ and note that $t = O(k)$, we get
\begin{align*}
\textstyle \lvert \cliqueL{\ell}{u} \rvert 
&= \textstyle \sum_{t = 1}^{k} \left[ \frac{(2k - t)t^{\ell - 3}}{(\ell - 3)!} + O(k^{\ell - 3}) \right]
\\&=\textstyle \frac{1}{(\ell - 3)!}\sum_{t = 1}^{k} (2kt^{\ell - 3} - t^{\ell - 2}) + O(k^{\ell - 2}).
\\ & = \textstyle \frac{1}{(\ell - 3)!}\left[2k \cdot \frac{k^{\ell - 2}}{\ell - 2} - \frac{k^{\ell - 1}}{\ell - 1}  \right] + O(k^{\ell - 2}),
\\ & = \textstyle \frac{\ell }{(\ell - 1)!} \cdot k^{\ell -1 } + O(k^{\ell - 2}) .
\end{align*}
\end{proof}

Proposition~\ref{prop:ccf_sw} shows that, when $p = 0$, $\accf{\ell}$ decreases as $\ell$ increases.
Furthermore, via simulation, we observe the same behavior as for $\accf{2}$ when
adjusting the rewiring probability $p$ (Fig.~\ref{fig:sw_ccfs}). Regardless of
$\ell$, the phase transition happens near $p = 0.1$.
Essentially, once there is enough rewiring, all local clique structure is lost, and 
clustering at all orders is lost.
This is partly a consequence of Proposition~\ref{prop:ccf_bounds}, which
says that  $\lccf{\ell}{u} \to 0$ as $\lccf{2}{u} \to 0$ for any $\ell$.


\section{Experimental results on real-world networks}\label{sec:empirical}

We now analyze the higher-order clustering of real-world networks.
We first study how the higher-order global and average clustering coefficients vary
as we increase the order $\ell$ of the clustering coefficient on a collection
of 20 networks from several domains.
After, we concentrate on a few representative networks and 
compare the higher-order clustering of real-world networks to null models. 
We find that only some networks exhibit higher-order clustering once the traditional
clustering coefficient is controlled.
Finally, we examine the local clustering of real-world networks.

\subsection{Higher-order global and average clustering}


\begin{table*}[t]
  \begin{tabular}{l r r @{\hskip 12pt} c c c @{\hskip 12pt} c c c @{\hskip 12pt} c c c}
    \toprule
    Network        & Nodes & Edges & $\gccf{2}$ & $\gccf{3}$ & $\gccf{4}$ & $\accf{2}$ & $\accf{3}$ & $\accf{4}$ & $\lvert \welldefnodes_2 \rvert / \lvert V \rvert$ & $\lvert \welldefnodes_3 \rvert / \lvert V \rvert$ & $\lvert \welldefnodes_4 \rvert / \lvert V \rvert$ \\ \midrule
    $\ER$ \cite{erdos1959random}                               & 1,000     & 99,831      & 0.200 & 0.040 & 0.008 & 0.200 & 0.040 & 0.008 & 1.000 & 1.000 & 1.000 \\
    $\smallworld$ \cite{watts1998collective}                   & 20,000    & 100,000     & 0.480 & 0.359 & 0.229 & 0.489 & 0.350 & 0.205 & 1.000 & 1.000 & 0.999 \\
    \midrule
    \emph{P.\ pacificus} \cite{bumbarger2013system}            & 50        & 576         & 0.015 & 0.051 & 0.035 & 0.073 & 0.052 & 0.034 & 0.880 & 0.580 & 0.440 \\
    \emph{C.\ elegans} \cite{watts1998collective}              & 297       & 2,148       & 0.181 & 0.080 & 0.056 & 0.308 & 0.137 & 0.062 & 0.949 & 0.926 & 0.808 \\
    Drosophila-medulla \cite{takemura2013visual}               & 1,781     & 32,311      & 0.000 & 0.002 & 0.001 & 0.116 & 0.061 & 0.024 & 0.803 & 0.616 & 0.425 \\
    mouse-retina \cite{helmstaedter2013connectomic}            & 1,076     & 577,350     & 0.008 & 0.038 & 0.029 & 0.033 & 0.100 & 0.085 & 0.998 & 0.996 & 0.994 \\
    \midrule
    fb-Stanford \cite{traud2012social}                         & 11,621    & 568,330     & 0.157 & 0.107 & 0.116 & 0.253 & 0.181 & 0.157 & 0.955 & 0.922 & 0.877 \\
    fb-Cornell \cite{traud2012social}                          & 18,660    & 790,777     & 0.136 & 0.106 & 0.121 & 0.225 & 0.169 & 0.148 & 0.973 & 0.951 & 0.923 \\
    Pokec \cite{takac2012data}                                 & 1,632,803 & 22,301,964  & 0.047 & 0.044 & 0.046 & 0.122 & 0.084 & 0.061 & 0.900 & 0.675 & 0.508 \\
    Orkut \cite{mislove2007socialnetworks}                     & 3,072,441 & 117,185,083 & 0.041 & 0.022 & 0.019 & 0.170 & 0.131 & 0.110 & 0.978 & 0.949 & 0.878 \\
    \midrule
    arxiv-HepPh \cite{leskovec2007graph}                       & 12,008    & 118,505     & 0.659 & 0.749 & 0.788 & 0.698 & 0.586 & 0.519 & 0.876 & 0.723 & 0.567 \\
    arxiv-AstroPh  \cite{leskovec2007graph}                    & 18,772    & 198,050     & 0.318 & 0.326 & 0.359 & 0.677 & 0.609 & 0.561 & 0.932 & 0.839 & 0.740 \\
    congress-committees \cite{porter2005committees}            & 871       & 248,848     & 0.037 & 0.080 & 0.063 & 0.082 & 0.142 & 0.126 & 1.000 & 1.000 & 1.000 \\
    DBLP \cite{yang2015defining}                               & 317,080   & 1,049,866   & 0.306 & 0.634 & 0.821 & 0.732 & 0.613 & 0.517 & 0.864 & 0.675 & 0.489 \\
    \midrule
    email-Enron-core \cite{Klimt-2004-Enron}                   & 148       & 1356        & 0.383 & 0.245 & 0.192 & 0.496 & 0.363 & 0.277 & 0.966 & 0.946 & 0.946 \\
    email-Eu-core \cite{yin2017local,leskovec2007graph}        & 1005      & 16064       & 0.267 & 0.170 & 0.135 & 0.450 & 0.329 & 0.264 & 0.887 & 0.847 & 0.784 \\
    CollegeMsg \cite{panzarasa2009patterns}                    & 1,899     & 41,579      & 0.004 & 0.005 & 0.003 & 0.053 & 0.017 & 0.006 & 0.829 & 0.591 & 0.332 \\
    wiki-Talk \cite{leskovec2010governance}                    & 2,394,385 & 4,659,565   & 0.002 & 0.011 & 0.010 & 0.201 & 0.081 & 0.051 & 0.262 & 0.077 & 0.027 \\
    \midrule
    oregon2-010526 \cite{leskovec2005graphs}                   & 11,461    & 32,730      & 0.037 & 0.085 & 0.097 & 0.494 & 0.294 & 0.300 & 0.711 & 0.269 & 0.121 \\
    as-caida-20071105 \cite{leskovec2005graphs}                & 26,475    & 53,381      & 0.007 & 0.012 & 0.015 & 0.333 & 0.159 & 0.134 & 0.625 & 0.171 & 0.060 \\
    p2p-Gnutella31 \cite{ripeanu2002mapping,leskovec2007graph} & 62,586    & 147,892     & 0.004 & 0.003 & 0.000 & 0.010 & 0.001 & 0.000 & 0.542 & 0.067 & 0.001 \\
    as-skitter \cite{leskovec2005graphs}                       & 1,696,415 & 11,095,298  & 0.005 & 0.007 & 0.011 & 0.296 & 0.126 & 0.109 & 0.871 & 0.633 & 0.335 \\
    \bottomrule
  \end{tabular}
  \caption{Higher-order clustering coefficients on random graph models, neural
    connections, online social networks, collaboration networks, human
    communication, and technological systems.  Broadly, networks from the same
    domain have similar higher-order clustering characteristics.  Since
    $\welldefnodes_\ell$ is the set of nodes at the center of at least one
    $\ell$-wedge (see Eq.~\ref{eq:def_accf}),
    $\lvert \welldefnodes_\ell \rvert / \lvert V \rvert$
    is the fraction of nodes at the center of at least one
    $\ell$-wedge (the higher-order average clustering coefficient
    $\accf{\ell}$ is only measured over those nodes participating in at least
    one $\ell$-wedge).}
  \label{tab:all_ccfs}
\end{table*}

We compute and analyze the higher-order clustering for networks from a variety
of domains (Table~\ref{tab:all_ccfs}).
We briefly describe the collection of networks and their categorization below:
\begin{enumerate}
\item Two synthetic networks---a random instance of an $\ER$ graph with
  $n=1,000$ nodes and edge probability $p=0.2$ and a small-world network with
  $n=20,000$ nodes, $k = 10$, and rewiring probability $p=0.1$;

\item Four neural networks---the complete neural systems of the nematode worms
  \emph{P.\ pacificus} and \emph{C.\ elegans} as well as the neural connections
  of the Drosophila medulla and mouse retina;

\item Four online social networks---two Facebook friendship networks between
  students at universities from 2005 (fb-Stanford, fb-Cornell) and two complete online
  friendship networks (Pokec and Orkut);

\item Four collaboration networks---two co-authorship networks constructed
  from arxiv submission categories (arxiv-AstroPh and arxiv-HepPh), 
  a co-authorship network constructed from DBLP, and the
  co-committee membership network of United States congresspersons (congress-committees);

\item Four human communication networks---two email networks (email-Enron-core,
  email-Eu-core), a Facebook-like messaging network from a college (CollegeMsg),
  and the edits of user talk pages by other users on Wikipedia (wiki-Talk); and

\item Four technological systems networks---three autonomous systems
  (oregon2-010526, as-caida-20071105, as-skitter) and a peer-to-peer connection
  network (p2p-Genutella31).
\end{enumerate}
In all cases, we take the edges as undirected, even if the original network data
is directed.

Table~\ref{tab:all_ccfs} lists the $\ell$th-order global and average clustering
coefficients for $\ell=2,3,4$ as well as the fraction of nodes that are the
center of at least one $\ell$-wedge (recall that the average clustering
coefficient is the mean only over higher-order local clustering coefficients of nodes
participating in at least one $\ell$-wedge; see \citet{kaiser2008mean} for a
discussion on how this can affect network analyses). We highlight some important
trends in the raw clustering coefficients, and in the next section, we focus on
higher-order clustering compared to what one gets in a null model.

Propositions~\ref{prop:ccf_er} and \ref{prop:ccf_sw} say that we should
expect the higher-order global and average clustering coefficients to
decrease as we increase the order $\ell$ for both the $\ER$ and small-world
models, and indeed $\accf{2} > \accf{3} > \accf{4}$ for these networks.  This
trend also holds for most of the real-world networks (mouse-retina,
congress-committees, and oregon2-010526 are the exceptions). Thus, when
averaging over nodes, higher-order cliques are overall less likely to close in
both the synthetic and real-world networks.




\newcommand{\realsigabove}{^{\ast}}
\newcommand{\fakesigabove}{\phantom{^{\ast}}}
\newcommand{\realsigbelow}{^{\dagger}}
\newcommand{\fakesigbelow}{\phantom{^{\dagger}}}

\begin{table*}[t]
\begin{tabular}{l c c c c c c c c c c c c c c c}
\toprule
& \multicolumn{3}{c}{\emph{C. elegans}} & \multicolumn{3}{c}{fb-Stanford} & \multicolumn{3}{c}{arxiv-AstroPh} & \multicolumn{3}{c}{email-Enron-core} & \multicolumn{3}{c}{oregon2-010526} \\
\cmidrule(lr){2-4}\cmidrule(lr){5-7}\cmidrule(lr){8-10}\cmidrule(lr){11-13}\cmidrule(lr){14-16}
               & original & CM         & MRCN       & original & CM         & MRCN       & original & CM         & MRCN       & original & CM & MRCN & original & CM & MRCN \\
\midrule
    $\accf{2}$ 
    & $0.31$ & $0.15\realsigabove$ & $0.31\fakesigabove$     
    & $0.25$ & $0.03\realsigabove$ & $0.25\fakesigabove$     
    & $0.68$ & $0.01\realsigabove$ & $0.68\fakesigabove$     
    & $0.50$ & $0.23\realsigabove$ & $0.50\fakesigabove$
    & $0.49$ & $0.25\realsigabove$ & $0.49\fakesigabove$
    \\
    $\accf{3}$ 
    & $0.14$ & $0.04\realsigabove$ & $0.17\realsigbelow$ 
    & $0.18$ & $0.00\realsigabove$ & $0.14\realsigabove$ 
    & $0.61$ & $0.00\realsigabove$ & $0.60\fakesigabove$    
    & $0.36$ & $0.08\realsigabove$ & $0.35\fakesigabove$ 
    & $0.29$ & $0.10\realsigabove$ & $0.14\realsigabove$  \\
  \bottomrule
  \end{tabular}
  \caption{Average higher-order clustering coefficients for five networks as
    well as the clustering with respect to two null models: a Configuration
    Model (CM) that samples random graphs with the same degree
    distribution~\cite{bollobas1980probabilistic,milo2003uniform}, and Maximally
    Random Clustered Networks (MRCN) that preserve degree distribution as well
    as $\accf{2}$~\cite{park2004statistical,colomer2013deciphering}.  For the
    random networks, we report the mean over 100 samples. An asterisk ($\ast$)
    denotes when the value in the original network is at least five standard
    deviations above the mean and a dagger ($\dagger$) denotes when the value
    in the original network is at least five standard deviations below the mean.
    Although all networks exhibit clustering with respect to CM,
    only some of the networks exhibit higher-order clustering when controlling
    for $\accf{2}$ with MRCN.}
  \label{tab:data_summary}
\end{table*}

The relationship between the higher-ordrer global clustering coefficient
$\gccf{\ell}$ and the order $\ell$ is less uniform over the datasets. For the three
co-authorship networks (arxiv-HepPh, arxiv-AstroPh, and DBLP) and the three
autonomous systems networks (oregon2-010526, as-caida-20071105, 
and as-skitter), $\gccf{\ell}$
increases with $\ell$, although the base clustering levels are much higher for
co-authorship networks. 
This is not simply due to the presence of cliques---a clique has the same clustering for any order
(Fig.~\ref{fig:ccf_diffs}, left). Instead, these datasets have nodes that serve
as the center of a star and also participate in a clique
(Fig.~\ref{fig:ccf_diffs}, right; see also Proposition~\ref{prop:ccf_bounds}).
On the other hand, $\gccf{\ell}$ decreases with $\ell$ for the two email
networks and the two nematode worm neural networks.  Finally, the change in
$\gccf{\ell}$ need not be monotonic in $\ell$.  In three of the four online
social networks, $\gccf{3} < \gccf{2}$ but $\gccf{4} > \gccf{3}$.

Overall, the trends in the higher-order clustering coefficients can be different
within one of our dataset categories, but tend to be uniform within
sub-categories: the change of $\accf{\ell}$ and $\gccf{\ell}$ with $\ell$
is the same for the two nematode worms within the neural networks,
the two email networks within the communication networks, and the three
co-authorship networks within the collaboration networks. These trends hold even
if the (classical) second-order clustering coefficients differ substantially in
absolute value.

While the raw clustering values are informative, it is also useful to compare
the clustering to what one expects from null models.
We find in the next section that this reveals additional insights into our data.

\subsection{Comparison against null models}

For one real-world network from each dataset category,
we also measure the higher-order clustering
coefficients with respect to two null models (Table~\ref{tab:data_summary}).
First, we compare against the Configuration Model (CM) that samples uniformly from
simple graphs with the same degree
distribution~\cite{bollobas1980probabilistic,milo2003uniform}.
In real-world networks, $\accf{2}$ is much larger than expected with respect to
the CM null model. We find that the same holds for $\accf{3}$.

Second, we use a null model that samples graphs preserving both degree
distribution and $\accf{2}$. Specifically, these are samples from an ensemble
of exponential graphs where the Hamiltonian measures the absolute value of the
difference between the original network and the sampled
network~\cite{park2004statistical}. Such samples are referred to as as
Maximally Random Clustered Networks (MRCN) and are sampled with a simulated
annealing procedure~\cite{colomer2013deciphering}. Comparing $\accf{3}$ between
the real-world and the null network, we observe different behavior in
higher-order clustering across our datasets. Compared to the MRCN null model, \emph{C. elegans} has
significantly less than expected higher-order clustering (in terms of
$\accf{3}$), the Facebook friendship and autonomous system networks have significantly
more than expected higher-order clustering,
and the co-authorship and email networks have
slightly (but not significantly) more than expected higher-order clustering
(Table~\ref{tab:data_summary}). Put another way, all real-world networks
exhibit clustering in the classical sense of triadic closure. However, the
higher-order clustering coefficients reveal that the friendship and autonomous
systems networks exhibit significant clustering beyond what is given by triadic 
closure. These results suggest the need for models that directly account for 
closure in node neighborhoods~\cite{bhat2016densification,lambiotte2016structural}.

Our finding about the lack of higher-order clustering in \emph{C. elegans}
agrees with previous results that 4-cliques are under-expressed, while open
3-wedges related to cooperative information propagation are
over-expressed~\cite{benson2016higher,milo2002network,varshney2011structural}.
This also provides credence for the ``3-layer'' model of
\emph{C. elegans}~\cite{varshney2011structural}. The observed clustering in the
friendship network is consistent with prior work showing the relative
infrequency of open $\ell$-wedges in many Facebook network subgraphs with
respect to a null model accounting for triadic
closure~\cite{ugander2013subgraph}. Co-authorship networks and email networks
are both constructed from ``events'' that create multiple edges---a paper with
$k$ authors induces a $k$-clique in the co-authorship graph and an email sent
from one address to $k$ others induces $k$ edges. This event-driven graph
construction creates enough closure structure so that the average third-order
clustering coefficient is not much larger than random graphs where the classical
second-order clustering coefficient and degree sequence is kept the same.

\begin{figure*}[tb]
\phantomsubfigure{fig:ccfsA}
\phantomsubfigure{fig:ccfsB}
\phantomsubfigure{fig:ccfsC}
\phantomsubfigure{fig:ccfsD}
\phantomsubfigure{fig:ccfsE}
\includegraphics[width=2\columnwidth]{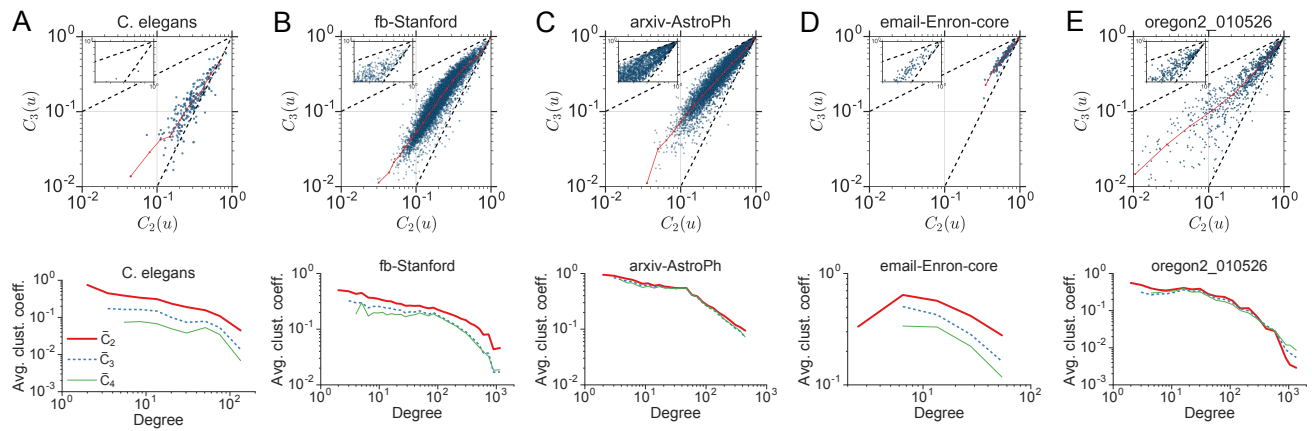}
\caption{Top row: Joint distributions of ($\lccf{2}{u}$,
  $\lccf{3}{u}$) for (A) \emph{C. elegans} (B) Facebook friendship,
  (C) arxiv co-authorship, (D) email, and
  (E) autonomous systems networks.
  Each blue dot represents a node, and the red curve
  tracks the average over logarithmic bins.  The upper trend line is the
  bound in Eq.~\ref{eq:Bound_Kappa3}, and the lower trend line is
  expected Erd\H{o}s-R\'enyi behavior from Proposition~\ref{prop:ccf_er_cond}.  
  Bottom row: Average higher-order clustering coefficients as a function of degree.
  }\label{fig:ccfs}
\end{figure*}

\begin{figure}[tb]
\phantomsubfigure{fig:ccfs_nullA}
\phantomsubfigure{fig:ccfs_nullB}
\includegraphics[width=0.9\columnwidth]{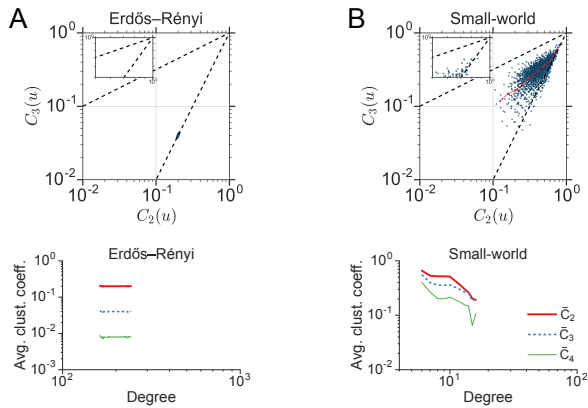}
\caption{Analogous plots of Fig.~\ref{fig:ccfs} for synthetic
   (A) $\ER$ and (B) small-world networks.
  Top row: Joint distributions of ($\lccf{2}{u}$, $\lccf{3}{u}$).  
  Bottom row: Average higher-order clustering coefficients as a function of degree.
  }\label{fig:ccfs_null}
\end{figure}

We emphasize that simple clique counts are not sufficient to obtain these
results. For example, the discrepancy in the third-order average clustering of
\emph{C. elegans} and the MRCN null model is not simply due to the
presence of 4-cliques. The original neural network has nearly twice as many
4-cliques (2,010) than the samples from the MRCN model (mean 1006.2, standard
deviation 73.6), but the third-order clustering coefficient is larger in MRCN.
The reason is that clustering coefficients normalize clique counts
with respect to opportunities for closure.

Thus far, we have analyzed global and average higher-order clustering, 
which both summarize the clustering of the entire network.
In the next section, we look at more localized properties, namely
the distribution of higher-order local clustering coefficients and the higher-order
average clustering coefficient as a function of node degree.

\subsection{Higher-order local clustering coefficients and degree dependencies}

We now examine more localized clustering properties of our networks.
Figure~\ref{fig:ccfs} (top) plots the joint distribution of $\lccf{2}{u}$ and
$\lccf{3}{u}$ for the five networks analyzed in Table~\ref{tab:data_summary},
and Fig.~\ref{fig:ccfs_null} (top) provides the analogous plots for the $\ER$
and small-world networks. In these plots, the lower dashed trend line represents the expected
$\ER$ behavior, i.e., the expected clustering if the edges in the neighborhood
of a node were configured randomly, as formalized in
Proposition~\ref{prop:ccf_er_cond}. The upper dashed trend line is the maximum
possible value of $\lccf{3}{u}$ given $\lccf{2}{u}$, as given by
Proposition~\ref{prop:ccf_bounds}.

For many nodes in \emph{C. elegans}, local clustering is nearly random
(Fig.~\ref{fig:ccfsA}, top), i.e., resembles the $\ER$ joint distribution
(Fig.~\ref{fig:ccfs_nullA}, top). In other words, there are many nodes that lie
on the lower trend line. This provides further evidence that \emph{ C. elegans}
lacks higher-order clustering. In the arxiv co-authorship network, there are
many nodes $u$ with a large value of $\lccf{2}{u}$ that have an even larger
value of $\lccf{3}{u}$ near the upper bound of Eq.~\ref{eq:Bound_Kappa3} (see
the inset of Fig.~\ref{fig:ccfsC}, top). This implies that some nodes appear in
both cliques and also as the center of star-like patterns, as in
Fig.~\ref{fig:ccf_diffs}. On the other hand, only a handful of nodes in the
Facebook friendships, Enron email, and Oregon autonomous systems networks are
close to the upper bound (insets of Figs.~\ref{fig:ccfsB},\ref{fig:ccfsD}, and
\ref{fig:ccfsE}, top).

Figures~\ref{fig:ccfs} and \ref{fig:ccfs_null} (bottom) plot higher-order
average clustering as a function of node degree in the real-world and synthetic
networks. In the $\ER$, small-world, \emph{C. elegans}, and Enron email
networks, there is a distinct gap between the average higher-order clustering
coefficients for nodes of all degrees. Thus, our previous finding that the
average clustering coefficient $\accf{\ell}$ decreases with $\ell$ in these
networks is independent of degree. In the Facebook friendship network,
$\lccf{2}{u}$ is larger than $\lccf{3}{u}$ and $\lccf{4}{u}$ on average for
nodes of all degrees, but $\lccf{3}{u}$ and $\lccf{4}{u}$ are roughly the same
for nodes of all degrees, which means that 4-cliques and 5-cliques close at
roughly the same rate, independent of degree, albeit at a smaller rate than
traditional triadic closure (Fig.~\ref{fig:ccfsB}, bottom). In the co-authorship
network, nodes $u$ have roughly the same $\lccf{\ell}{u}$ for $\ell = 2$, $3$,
$4$, which means that $\ell$-cliques close at about the same rate, independent
of $\ell$ (Fig.~\ref{fig:ccfsC}, bottom). In the Oregon autonomous systems
network, we see that, on average, $\lccf{4}{u} > \lccf{3}{u} > \lccf{2}{u}$ for
nodes with large degree (Fig.~\ref{fig:ccfsE}, bottom). This explains how the
global clustering coefficient increases with the order, but the average
clustering does not, as observed in Table~\ref{tab:all_ccfs}.


\section{Discussion}

We have proposed higher-order clustering coefficients to study
higher-order closure patterns in networks, which generalizes the widely used
clustering coefficient that measures triadic closure.
Our work compliments other recent developments on the importance of higher-order
information in network navigation~\cite{rosvall2014memory,scholtes2017network}
and on temporal community structure~\cite{sekara2016fundamental}; in contrast,
we examine higher-order clique closure and only implicitly consider time as a
motivation for closure.

Prior efforts in generalizing clustering coefficients have focused on shortest
paths~\cite{fronczak2002higher}, cycle formation~\cite{caldarelli2004structure},
and triangle frequency in $k$-hop
neighborhoods~\cite{andrade2006neighborhood,jiang2004topological}.
Such approaches fail to capture closure patterns of cliques, suffer from
challenging computational issues, and are difficult to theoretically analyze in random graph
models more sophisticated than the Erd\H{o}s-R\'enyi model.
On the other hand, our higher-order clustering coefficients are simple but
effective measurements that are analyzable and easily computable (we only rely
clique enumeration, a well-studied algorithmic task).
Furthermore, our methodology provides new insights into the clustering behavior
of several real-world networks and random graph models, and our theoretical
analysis provides intuition for the way in which higher-order clustering
coefficients describe local clustering in graphs.

Finally, we focused on higher-order clustering coefficients as a global network
measurement and as a node-level measurement, and in related work we also
show that large higher-order clustering implies the existence of mesoscale clique-dense
community structure~\cite{yin2017local}.

\begin{acknowledgments}
This research has been supported in part by NSF
IIS-1149837, ARO MURI, DARPA, ONR,
Huawei, 
and Stanford Data Science Initiative.
We thank Will Hamilton and Marinka \v{Z}itnik for
insightful comments.
We thank Mason Porter and Peter Mucha for providing the congress committee membership data.
\end{acknowledgments}

\bibliographystyle{apsrev4-1}
\bibliography{refs}

\end{document}